\documentclass[conference]{IEEEtran}
\IEEEoverridecommandlockouts

\usepackage{cite}
\usepackage{amsmath,amssymb,amsfonts,amsthm}
\usepackage{commath}
\usepackage{mathtools}
\usepackage{color}
\usepackage{algorithmic}
\usepackage{graphicx}
\usepackage{textcomp}
\usepackage{xcolor}
\def\BibTeX{{\rm B\kern-.05em{\sc i\kern-.025em b}\kern-.08em
    T\kern-.1667em\lower.7ex\hbox{E}\kern-.125emX}}

\newcommand{\0}{\pmb{0}}

\newcommand{\Um}{\pmb{U}}

\newcommand{\CN}{\mathcal{CN}}
\newcommand{\E}{\mathbb{E}}

\newcommand{\gv}{\pmb{g}}

\newcommand{\yv}{\pmb{y}}
\newcommand{\uv}{\pmb{u}}
\newcommand{\vv}{\pmb{v}}
\newcommand{\xv}{\pmb{x}}
\newcommand{\zv}{\pmb{z}}
\newcommand{\wv}{\pmb{w}}

\newcommand{\Cm}{\pmb{C}}
\newcommand{\Hm}{\pmb{H}}
\newcommand{\Gm}{\pmb{G}}
\newcommand{\Fm}{\pmb{F}}
\newcommand{\Am}{\pmb{A}}
\newcommand{\Bm}{\pmb{B}}
\newcommand{\Dm}{\pmb{D}}

\newcommand{\Rm}{\pmb{R}}
\newcommand{\Sm}{\pmb{S}}

\newcommand{\Vm}{\pmb{V}}
\newcommand{\Qm}{\pmb{Q}}

\newcommand{\Id}{\pmb{I}}
\newtheorem{proposition}{Proposition}

\newtheorem{lemma}{Lemma}

\newcommand{\rank}{\operatorname{rank}}

\renewcommand{\P}{\mathrm{P}}
\newcommand{\Scal}{\mathcal{S}}
\newcommand{\Xcal}{\mathcal{X}}

\newif\ifdraft

\ifdraft
  \newcommand{\NEW}[1]{\textcolor{red}{#1}}
\else
  \newcommand{\NEW}[1]{#1}
\fi

\DeclareMathOperator{\tr}{tr}

\begin{document}

\title{Capacity Bounds on Doppler OFDM Channels}
\author{\IEEEauthorblockN{Pablo Orellana$\stackrel{\ddag}{}$$\stackrel{^\S}{}$, Zheng Li$\stackrel{\ddag}{}$, 
Jean-Marc Kelif$\stackrel{\ddag}{}$, Sheng Yang$\stackrel{^\S}{}$, and Shlomo Shamai~(Shitz)$\stackrel{^*}{}$}
\IEEEauthorblockA{$^\ddag$ Orange Innovation, 92320, Chatillon, France\\
$^\S$Universit\'e Paris-Saclay, CNRS, CentraleSup\'elec, Laboratoire des signaux et systèmes, 91190, Gif-sur-Yvette, France\\
$^*$Technion, Haifa, 3200003, Israel\\
Email: \{pablo.orellana, zheng1.li, jeanmarc.kelif\}@orange.com,  sheng.yang@centralesupelec.fr, sshlomo@ee.technion.ac.il}}

\maketitle

\begin{abstract}
Low Earth orbit (LEO) satellite systems experience significant Doppler effects due to high mobility. 
While Doppler shifts can be largely compensated, residual frequency uncertainty induces a structured 
form of channel uncertainty that can limit achievable rates.
We model this effect using a block-fading channel of the form $\Hm = \Fm + s\Gm$, where $s$ is an unknown scalar random parameter. We first study this model in a general $N\times N$ MIMO setting. For this channel, we derive achievable rate lower bounds based on explicit transmission schemes and capacity  upper bounds using a duality approach.
We study Gaussian signaling and propose a practical superposition scheme with subspace alignment (SN) and successive interference cancellation, where a coarse-layer stream serves as an implicit pilot for decoding refined-layer data. We characterize asymptotic capacity in the near-coherent and high-SNR regimes, and show via Doppler-OFDM simulations that the proposed SN scheme achieves near-optimal rates with low complexity.
\end{abstract}

\section{Introduction}

While 5G enables high throughput, low latency, and massive connectivity, ubiquitous coverage remains challenging, particularly in remote areas. Non-terrestrial networks (NTNs), especially low Earth orbit (LEO) satellite systems, provide a promising complement but are impaired by high path loss, long delays, and strong Doppler effects. The high orbital velocity of LEO satellites induces rapid frequency variations, causing synchronization errors and inter-carrier interference in OFDM systems and degrading achievable rates. Understanding Doppler effects is therefore critical for NTN design.

Although Doppler effects in NTNs have been widely studied, information-theoretic capacity analyses remain limited. Existing works either assume ideal Doppler estimation and compensation~\cite{goto2018leo,sedin2020throughput}, or focus on Doppler estimation and mitigation techniques~\cite{lin2021doppler,pan2020efficient,yeh2024efficient}, treating Doppler primarily as an impairment rather than a source of fundamental channel uncertainty. Moreover, the lack of channel models that are both analytically tractable and Doppler-faithful further hinders capacity analysis: statistical models rely on stationary assumptions~\cite{ngo2018performance}, while accurate linear time-varying models are generally intractable.

From an information-theoretic perspective, channel uncertainty has long been known to induce behaviors markedly different from the coherent Gaussian case. On the converse side, the duality framework of Lapidoth and Moser provides sharp capacity upper bounds for noncoherent flat-fading channels~\cite{LapidothMoser2003Duality}. On the achievability side, geometric signaling schemes for noncoherent MIMO, such as unitary space--time modulation and Grassmann-manifold coding, characterize the degrees-of-freedom penalty due to unknown fading~\cite{ZhengTse2002Grassmann}. A recurring theme is that uncertainty often leads to non-Gaussian optimal inputs: for noncoherent Rayleigh fading, the capacity-achieving input is discrete with finitely many mass points~\cite{AbouFaycalTrottShamai2001RayleighDiscrete}, and related discrete-structure results extend to Rician models under peakedness constraints~\cite{GursoyPoorVerdu2005RicianPartI}. Phase uncertainty provides another canonical example, where high-SNR capacity characterizations again rely on duality-type arguments and differ fundamentally from coherent AWGN~\cite{Lapidoth2002PhaseNoiseHighSNR, YangShamai17}. These results motivate the present work, which studies a structured Doppler-induced uncertainty model using matching achievability schemes and duality-based upper bounds.

In this paper, we model the residual Doppler effect using a tractable block-fading channel of the form $\Hm = \Fm + s\Gm$, where $s$ is an unknown scalar random parameter. We first study this model in a general $N\times N$ MIMO setting. For this channel, we derive achievable rate lower bounds using explicit transmission schemes and capacity upper bounds via a duality-based approach, and show that rank-one uncertainty limits the system to at most $N-1$ degrees of freedom. On the achievability side, we analyze Gaussian signaling and propose a superposition-based precoding scheme with subspace alignment and successive interference cancellation, where one stream from the coarse layer serves as an implicit pilot to enable decoding of the remaining streams from the refined layer. We also characterize the near-coherent regime with small uncertainty variance. Finally, we apply the results numerically to Doppler-affected OFDM channels in LEO satellite systems.  

The remainder of this paper is organized as follows. Section~II introduces the system model and problem formulation. Section~III presents a lower bound based on Gaussian signaling. The proposed scheme is described and analyzed in  Section~IV. Section~V derives capacity upper bounds and asymptotic behaviors. Section~VI  applies the results to Doppler-affected OFDM channels, and concludes the paper.

\section{Notation}
Scalars are denoted by lowercase letters, e.g., $s$, vectors by bold lowercase
letters, e.g., $\xv,\yv,\zv$, and matrices by bold uppercase letters
(e.g., $\Fm,\Gm,\Hm,\Qm_x$). The $N\times N$ identity matrix is $\Id_N$. We denote $(\cdot)^T$ and $(\cdot)^H$ as the transpose and the Hermitian (conjugate) transpose, respectively.
The trace, determinant, and rank operators are $\tr(\cdot)$, $\det(\cdot)$,
and $\rank(\cdot)$. Expectations are written $\E[\cdot]$. Random variables/vectors following a circularly symmetric complex Gaussian
distribution are written $\mathcal{CN}(\pmb{\mu},\pmb{\Sigma})$. Unless stated
otherwise, additive noise is $\zv\sim\mathcal{CN}(\0,\Id_N)$. The Euclidean
norm is $\|\cdot\|$; more generally, $\|\yv\|_{\Sm}^2 := \yv^H\Sm\yv$ denotes a
quadratic form induced by $\Sm\succeq 0$. The indicator function is
$\pmb{1}(\cdot)$, and $\delta[\cdot]$ denotes the Kronecker delta. The Frobenius norm is $\| \cdot \|_F$. The symbol $(\cdot)^\perp$ denotes the orthogonal complement of a subspace.

\section{System Model and Problem Formulation}
\label{sec:model}

\subsection{Motivating example: OFDM channel with Doppler effects}
\label{subsec:ofdm_motivation}

We briefly recall a standard OFDM model to motivate the abstract channel studied
in this paper. Consider a discrete-time $L$-tap multipath channel
\begin{equation*}
  y[m] = e^{j2\pi \phi[m]} \sum_{l=0}^{L-1} h_l x[m-l] + z[m],
  \quad m=1,\ldots,N+N_{\mathrm{cp}},
\end{equation*}
where $N_{\mathrm{cp}}=L-1$ is the cyclic prefix length, $\{h_l\}$ are the channel
taps, and $z[m]$ is additive white Gaussian noise.
The phase rotation $\phi[m]$ models Doppler effects and is given by
$\phi[m] = \frac{m}{N} f_d$, 
where $f_d$ denotes the Doppler shift normalized by the OFDM subcarrier spacing.
We assume that $f_d$ is random but constant over one OFDM symbol.

After cyclic prefix removal and discrete Fourier transform~(DFT), the
frequency-domain input--output relation can be written as
\begin{equation}
  \tilde{\yv} = \Hm \tilde{\xv} + \tilde{\zv},
\end{equation}
where $\tilde{\xv}$, $\tilde{\yv}$, and $\tilde{\zv}$ denote the frequency-domain
transmit signal, received signal, and noise vector, respectively, and
$\Hm\in\mathbb C^{N\times N}$ is the effective Doppler-OFDM frequency-domain channel matrix \NEW{given by:}
\begin{align}
  &\Hm(i,k) = \biggl(\sum_{l=0}^{L-1} h_l e^{-j2\pi \frac{l}{N} k} \biggr) \cdot \Bm (i,k), \\
  & \Bm(i,k) = \frac{1}{N} \frac{\sin (\pi(f_d + k-i))}{\sin \left(\frac{\pi}{N}(f_d + k-i)\right)} e^{j 2 \pi \Psi(i,k) }, \\
  & \Psi(i,k) = \frac{1}{2} \left[ \left( \frac{2L-1}{N} +1 \right) f_d +\left( 1-\frac{1}{N} \right)(k-i) \right],
\end{align}
\NEW{where $\Bm \in\mathbb C^{N\times N}$ represents the inter-carrier interference~(ICI) caused by the Doppler shift.} In the absence of Doppler ($f_d=0$, e.g., after perfect compensation), $\Hm$ is diagonal; Doppler induces inter-carrier interference through off-diagonal terms.

When the Doppler shift is small, a first-order Taylor expansion of $\Hm(f_d)$
around $f_d=0$ (or, more generally, around a Doppler estimate $\hat f_d$) yields
the approximation
\begin{equation}
  \Hm \approx \Fm + f_d \Gm,
  \label{eq:doppler_linear}
\end{equation}
\NEW{
where for the case of $f_d$ having a small variance}
\begin{align}
  \Fm(i,k) &:= \sum_{l=0}^{L-1} h_l e^{-j2\pi \frac{l}{N} k} \delta[i-k], \\
  \Gm(i,k) &:= \biggl(\sum_{l=0}^{L-1} h_l e^{-j2\pi \frac{l}{N} k} \biggr) \Dm(i,k),
\end{align}%
\begin{align}
    \Dm(i,k) & := \dfrac{\pi}{N}\,\dfrac{e^{-j\tfrac{\pi}{N}(k-i)}}{\sin\!\left(\tfrac{\pi}{N}(k-i)\right)} \pmb{1}(k\ne i) \nonumber \\ & \qquad + j\pi\!\left(1+\dfrac{2L-1}{N}\right) \pmb{1}(k=i),
\end{align}
where $\Fm$ corresponds to the nominal frequency-domain channel and $\Gm$
captures the sensitivity to Doppler-induced inter-carrier interference.

If Doppler is imperfectly compensated but known up to a residual estimation
error, then $\Fm$ is no longer diagonal but remains deterministic and known at
the receiver, while $f_d$ in \eqref{eq:doppler_linear} represents the residual
Doppler error with small variance.
This approximation motivates the following abstract channel model.

\subsection{Abstract channel model}
\label{subsec:abstract_model}

From \eqref{eq:doppler_linear}, we consider the following memoryless channel:
\begin{equation}
  \yv[i] = \Hm[i]\xv[i] + \zv[i], \qquad i=1,\ldots,n,
\end{equation}
where $\xv[i]\in\mathbb C^{N}$ is the channel input and
$\zv[i]\sim\mathcal{CN}(\mathbf 0,\Id_N)$ is additive white Gaussian noise; 
the channel matrix is 
\begin{equation}
  \Hm[i] = \Fm + s[i]\Gm,
  \label{eq:abstract_model}
\end{equation}
where $\Fm,\Gm\in\mathbb C^{N\times N}$ are fixed and known to both transmitter and
receiver, and $\{s[i]\}$ are i.i.d.\ $\mathcal{CN}(0,\sigma^2)$ random variables
modeling channel uncertainty.
Unless stated otherwise, $\Fm$ and $\Gm$ may have any structure.
The input is subject to an average power constraint
\begin{equation}
  \frac{1}{n}\sum_{i=1}^n \mathbb E\|\xv[i]\|^2 \le P.
\end{equation}
Since the channel is stationary and memoryless, its capacity admits the
single-letter characterization
\begin{equation}
  C(P) = \max_{P_X:\,\mathbb E\|\xv\|^2\le P} I(\xv;\yv).
\end{equation}
Conditioned on $\xv$, the output $\yv$ is complex Gaussian with mean
$\pmb{\mu}_{\yv|\xv} = \Fm\xv$
and covariance
\begin{equation}
  \pmb{\Sigma}_{\yv|\xv}
  =
  \Id_N + \sigma^2\,\Gm\xv\xv^H\Gm^H.
\end{equation}
Therefore, the conditional differential entropy is
\begin{equation}
  h(\yv\mid\xv)
  =
  N\log(\pi e)
  + \mathbb E\!\left[
    \log\!\left(1+\sigma^2\|\Gm\xv\|^2\right)
  \right]. \label{eq:h(y|x)}
\end{equation}

Throughout the paper, logarithms are natural and rates are measured in nats.

\section{Achievability with Gaussian Signaling}
\label{sec:gaussian_ach}

We first study achievability using single-layer Gaussian signaling, i.e., let the channel input be distributed as $\xv \sim \mathcal{CN}(\mathbf 0,\Qm_x)$, 
where $\Qm_x\succeq 0$ and $\tr(\Qm_x)\le P$.

\subsection{Optimal decoding}
\label{subsec:gaussian_opt}
Conditioned on a realization $s=s_0$, the channel output is complex Gaussian with zero mean and covariance
\begin{equation}
  \pmb\Sigma_{\yv|s=s_0}
  :=
  \Id + (\Fm+s_0\Gm)\Qm_x(\Fm+s_0\Gm)^H.
  \label{eq:Sigma_y_given_s}
\end{equation}
Hence, $\P_{\yv|s=s_0} = \mathcal{CN}\!\left(\mathbf 0,\;\pmb\Sigma_{\yv|s=s_0}\right)$, and 
the conditional differential entropy is 
\begin{align}
  h(\yv\mid s)
  &=
  N\log(\pi e)
  + \mathbb E_s \!\left[
    \log\det\!\big(\pmb\Sigma_{\yv|s}\big)
  \right].
  \label{eq:hy_given_s_gaussian}
\end{align}
Since conditioning reduces entropy, 
\begin{equation}
  I(\xv;\yv)
  = h(\yv)-h(\yv\mid\xv)
  \ge h(\yv\mid s)-h(\yv\mid\xv).
\end{equation}
Using the conditional covariance structure of $\yv\mid\xv$, this yields
\begin{align}
  I(\xv;\yv)
  &\ge
  \mathbb E_s \!\left[
    \log\det\!\big(\pmb\Sigma_{\yv|s}\big)
  \right]
  - \mathbb E_{\xv}\!\left[
    \log\!\left(1+\sigma^2\,\|\Gm\xv\|^2\right)
  \right].
  \label{eq:gaussian_lb}
\end{align}

Applying Jensen’s inequality to the second term leads to the following result.

\begin{proposition}[Gaussian signaling with optimal decoding]
\label{prop:gaussian_lb}
The channel capacity satisfies
\begin{equation}
  C(P)
  \ge
  \max_{\Qm_x\succeq0:\,\tr(\Qm_x)\le P}
  R_G(\Qm_x),
  \label{eq:RG}
\end{equation}
where
\begin{multline}
  R_G(\Qm_x)
  := 
  \E_s\!\left[
    \log\det\!\big(\Id + (\Fm+s\Gm)\Qm_x(\Fm+s\Gm)^H\big)
  \right]
  \\ -
  \log\!\big(1+\sigma^2\tr(\Gm\Qm_x\Gm^H)\big).
\end{multline}
\end{proposition}
The lower bound \eqref{eq:RG} is tight and coincides with the channel capacity
when $\sigma^2=0$, in which case the channel reduces to a deterministic MIMO channel with matrix $\Fm$.

\subsection{Linear receiver and GMI bound}
\label{subsec:gaussian_lmmse}
Instead of applying the optimal receiver, one can also use a suboptimal linear receiver, which corresponds to the following lower bound.  
\begin{align}
I(\xv;\yv) &= h(\xv) - h(\xv|\yv) \\
&= h(\xv) - h(\xv - \hat{\xv}(\yv)|\yv) \label{eq:shift_inv_giv_y}\\
&\ge h(\xv) - h(\xv - \hat{\xv}(\yv)) \label{eq:cond_red_entropy}\\
&\ge \log\det(\Qm_x) - \log\det(\Qm_{x|y}^{\mathrm{LMMSE}}), \label{eq:lin_rx_lb}
\end{align}
\NEW{where \eqref{eq:shift_inv_giv_y} follows from shift invariance of differential entropy given a random variable, \eqref{eq:cond_red_entropy} is from conditioning reduces entropy and \eqref{eq:lin_rx_lb} is from the fact that Gaussian random vector maximizes differential entropy for a fixed covariance. We choose the LMMSE estimator as the receiver, with error covariance}
\begin{equation}
  \Qm_{x|y}^{\mathrm{LMMSE}}
  :=
  \E\!\left[(\xv-\Am(\yv)\yv)(\xv-\Am(\yv)\yv)^H\right],
  \label{eq:NN_metric_LMMSE}
\end{equation}
with $\Am(\yv)$ the LMMSE estimator matrix matched to the second-order statistics.
\NEW{Then, the output covariance $\Qm_y$ and input--output cross-covariance $\Qm_{x y}$ are given by
\begin{align*}
    \Qm_y & =  \E[\yv \yv^H] = \Fm \Qm_x \Fm^H + \Rm_0,
\end{align*}
and
\begin{align*}
    \Qm_{xy} & = \E[\xv \yv^H] = \Qm_x \Fm^H.
\end{align*}
The LMMSE estimator is therefore
\begin{align}
    \Am & = \Qm_{xy} \Qm_y^{-1} \\
    & = \Qm_x \Fm^H (\Fm \Qm_x \Fm^H + \Rm_0)^{-1},
    \label{eq:lmmse_estimator}
\end{align}
}after \NEW{replacing \eqref{eq:lmmse_estimator} in \eqref{eq:NN_metric_LMMSE} and doing} some algebra, we can obtain
\begin{equation}
  \Qm_{x\mid y}^{\mathrm{LMMSE}}
  =
  \big(\Qm_x^{-1}+\Fm^H\Rm_0^{-1}\Fm\big)^{-1},
  \label{eq:LMMSE_error_cov_single}
\end{equation}
where the equivalent interference covariance matrix $\Rm_0$ is
\begin{equation}
    \Rm_0 := \Id + \sigma^2\,\Gm\Qm_x\Gm^H.
    \label{eq:R_0}
\end{equation}
Then, we get the following lower bound with Gaussian input and linear receiver. 
\begin{proposition}[Gaussian signaling with linear receiver]
\label{prop:gaussian_lb_lmmse}
The channel capacity satisfies
\begin{equation}
  C(P)
  \ge
  \max_{\Qm_x\succeq0:\,\tr(\Qm_x)\le P}
  R_G^{\mathrm{Lin}}(\Qm_x),
  \label{eq:RG_lin}
\end{equation}
where
\begin{equation}
  R_G^{\mathrm{Lin}}(\Qm_x)
  :=
  \log\det\!\big(
    \Id + \Qm_x\Fm^H(\Id+\sigma^2\Gm\Qm_x\Gm^H)^{-1}\Fm
  \big).
  \label{eq:RG_lin_single}
\end{equation}
\end{proposition}

The bound~\eqref{eq:RG_lin} is an instance of the
\emph{generalized mutual information} (GMI)~\cite{lapidoth1996mismatch}, i.e., the
rate achievable with a Gaussian codebook and a nearest-neighbor decoding metric
matched to the LMMSE error covariance~$\Qm_{x\mid y}^{\mathrm{LMMSE}}$.
The bound is tight when $\sigma^2=0$, since the mismatched metric then coincides
with the matched decoder.

While the linear receiver is simple and practical, it treats the self interference as unstructured noise, resulting in an effective disturbance covariance
$\Rm_0$ that can dominate at high signal power or large uncertainty~$\sigma^2$.
In contrast, optimal decoding exploits the structure of the channel uncertainty but leads to decoding metrics that are nonconvex and difficult to implement in practice.
This gap motivates the superposition and subspace-alignment scheme developed in the next section.

\section{Superposition Coding and Subspace Alignment} 

We consider a superposition coding scheme with two independent layers,
\begin{equation}
  \xv = \xv_d + \xv_p,
\end{equation}
where $\xv_p$ carries a coarse layer decoded first and acting as an implicit
pilot, while $\xv_d$ carries a refined data layer, inducing successive
interference cancellation (SIC).
The refined layer is linearly precoded as
\begin{equation}
  \xv_d = \Vm \wv_d,
\end{equation}
where $\wv_d\in\mathbb{C}^{M}$ with $M\le N$, and
$\Vm\in\mathbb{C}^{N\times M}$ is semi-unitary ($\Vm^H\Vm=\Id$) and depends only on
known channel parameters.
The single-layer Gaussian scheme is recovered by setting $\xv_p=0$ and $M=N$,
while a pilot-based scheme corresponds to a deterministic nonzero $\xv_p$.
The achievable rate decomposes as
\begin{equation}
  I(\wv_d,\xv_p;\yv)
  =
  I(\xv_p;\yv)+I(\wv_d;\yv\mid\xv_p).
\end{equation}
We design $\Vm$ so that the refined layer aligns with a receiver subspace
invariant to the unknown coefficient $s$, enabling $\xv_p$ to be decoded without
channel knowledge and then used to decode $\wv_d$.

\subsection{Subspace alignment precoding}

\begin{lemma}
\label{lemma:FG}
Let $\Fm,\Gm\in\mathbb{C}^{N\times N}$. There exists a semi-unitary matrix
$\Vm\in\mathbb{C}^{N\times(N-1)}$, with $\Vm^{H}\Vm=\Id_{N-1}$, such that
\[
\rank\!\big[\,\Fm \Vm\ \ \Gm \Vm\,\big]\le N-1.
\]
We denote by $\mathcal{V}:= \mathcal{V}(\Fm,\Gm)$ the set of all such matrices $\Vm$.
\end{lemma}

\begin{proof}
We show that one can always find a nonzero $\wv\in\mathbb{C}^{N}$ and a semiunitary
$\Vm$ with columns in $\ker(\wv^{H}\Fm)\cap\ker(\wv^{H}\Gm)$.
Then $\wv^{H}\Fm \Vm=\0$ and $\wv^{H}\Gm \Vm=0$, so every column of $[\Fm\Vm\ \ \Gm\Vm]$
lies in $\wv^{\perp}$, hence $\rank[\Fm\Vm\ \ \Gm\Vm]\le \dim(\wv^{\perp})=N-1$.

\emph{Case 1: $\rank(\Fm)<N$.}
Pick $\wv\in\ker(\Fm^{H})\setminus\{0\}$, so $\wv^{H}\Fm=0$.
If $\wv^H \Gm = 0$, then $\wv^H\Gm \Vm = 0$ for all $\Vm$. 
Otherwise, $\ker(\wv^{H}\Gm)$ has codimension $1$, thus dimension $N-1$.
Choose $\Vm$ whose columns form an orthonormal basis of $\ker(\wv^{H}\Gm)$.
By construction $\wv^{H}\Fm \Vm=0$ and $\wv^{H}\Gm \Vm=0$.

\emph{Case 2: $\rank(\Gm)<N$.}
Symmetrically, pick $\wv\in\ker(\Gm^{H})\setminus\{0\}$ and let $\Vm$ have
orthonormal columns spanning $\ker(\wv^{H}\Fm)$ (dimension $N-1$). Then
$\wv^{H}\Gm \Vm=0$ and $\wv^{H}\Fm \Vm=0$.

\emph{Case 3: $\rank(\Fm)=\rank(\Gm)=N$.}
Consider the matrix pencil $\Fm+t\Gm$ and the polynomial $p(t)=\det(\Fm+t\Gm)$,
which has degree $N$ with nonzero leading and constant coefficients
($\det(\Gm)\neq0$, $\det(\Fm)\neq0$). By the fundamental theorem of algebra,
$p$ has a root $t_\star\in\mathbb{C}$, so $\Fm+t_\star \Gm$ is singular.
Let $\wv\neq0$ satisfy $(\Fm+t_\star \Gm)^{H}\wv=0$, i.e.,
$\Fm^{H}\wv=-t_\star^{*}\Gm^{H}\wv$, so $\Fm^{H}\wv$ and $\Gm^{H}\wv$ are colinear.
Hence $\ker(\wv^{H}\Fm)=\ker(\wv^{H}\Gm)$ is a hyperplane of dimension $N-1$.
Choose $\Vm$ with orthonormal columns spanning this common kernel. Then
$\wv^{H}\Fm \Vm=0$ and $\wv^{H}\Gm \Vm=0$.

In all cases such a semiunitary $\Vm$ exists and satisfies
$\rank[\Fm\Vm\ \ \Gm\Vm]\le N-1$.
\end{proof}
For convenience, define
\[
\Scal(\Vm):=\Scal(\Fm,\Gm,\Vm)
:=
\mathrm{span}  \big(\Fm\Vm,\Gm\Vm\big),
\]
and $\Scal^{\perp}(\Vm):=\Scal^{\perp}(\Fm,\Gm,\Vm)$ the null space with dimension $d_\perp$;
and let $\Um$ and $\Um_\perp$ be semi-unitary matrices whose columns form
orthonormal bases of $\Scal$ and $\Scal^{\perp}$, respectively.

The refined-layer signal is precoded as $\xv_d=\Vm\wv_d$, where
$\wv_d\sim\CN(0,\Qm_d)$ with $\Qm_d\succeq0$ and $P_d:=\tr(\Qm_d)\le P$,
and $\Vm\in\mathcal{V}(\Fm,\Gm)$ is the precoder that can be optimized.
It then follows from Lemma~\ref{lemma:FG} that, for any $s$ and $\wv_d$, the columns of
  $\Hm\xv_d = (\Fm\Vm+s\Gm\Vm)\wv_d$ belong to  $\Scal(\Fm,\Gm,\Vm)$,
that is, the refined layer is confined to a fixed receiver subspace of dimension
at most $N-1$, independent of the realization of $s$.

\subsection{Coarse layer encoding and decoding}

The received signal is $\yv = \Hm \xv_p + \Hm \xv_d + \zv$.
By the data processing inequality,
\begin{align}
  I(\xv_p;\yv)
  &\ge I(\xv_p;\Um_\perp^H\yv)
   = I(\xv_p;\Um_\perp^H\Hm\xv_p+\Um_\perp^H\zv).
\end{align}
Choose a fixed unit vector $\vv_p$ and 
$$\xv_p=\vv_p w_p,$$ 
where $w_p\sim\P_{w_p}$ satisfies
$\E|w_p|^2\le P_p:= P-P_d$. Then, the achievable rate of the coarse layer is 
\begin{equation*}
  R_p(\P_{w_p},\Vm,\vv_p)
  :=
  I\!\left(
    w_p;\,
    (\Um_\perp^H\Fm \vv_p+s\,\Um_\perp^H\Gm  \vv_p)w_p+\zv_p
  \right),
\end{equation*}
where $\zv_p:=\Um_\perp^H\zv\sim\CN(0,\Id)$; The rate depends on $\Um_\perp$ only through the subspace $\Scal(\Vm)$ because the Gaussian noise is isotropic.
Taking any column of $\Um_\perp$, say, $\uv_\perp$, we obtain a scalar Ricean channel, with deterministic component
$\uv_\perp^H\Fm\vv_p$ and random component
$s\,\uv_\perp^H\Gm\vv_p\sim\CN\!\big(0,|\uv_\perp^H\Gm\vv_p|^2\big)$.

\subsection{Estimation of $s$ from $\xv_p$ via orthogonal projection}
\label{subsec:s_est2}

After decoding the coarse layer, we subtract its deterministic contribution and
define
\begin{equation}
  \tilde{\yv}
  :=
  \yv-\Fm\xv_p
  =
  s\,\Gm\xv_p + (\Fm+s\Gm)\Vm\wv_d + \zv .
\end{equation}
To eliminate the refined-layer signal, we project $\tilde{\yv}$ onto the
orthogonal complement
$\Scal^{\perp}$.
Define
\begin{equation}
  {\tilde{\yv}_\perp}
  :=
  \Um_\perp^H\tilde{\yv}
  =
  s\,\Um_\perp^H\Gm\xv_p + {\zv}_\perp,
  \qquad
  {\zv}_\perp\sim\CN(0,\Id_{d_\perp}),
\end{equation}
where $\Um_\perp^H\Fm\Vm=\Um_\perp^H\Gm\Vm=0$ by construction.
Thus, the refined-layer contribution is completely removed, and ${\tilde\yv}_\perp$
depends only on $(s,\xv_p)$ and noise.

Assuming $s\sim\CN(0,\sigma^2)$, the pair $(s,{\tilde\yv}_\perp)$ is jointly Gaussian, and
the MMSE estimator of $s$ from ${\tilde\yv}_\perp$ coincides with the LMMSE estimator,
\begin{equation}
  \hat{s}
  ={(\|\gv_\perp\|^2+\sigma^{-2}})^{-1}
  {\gv_\perp^H{\tilde\yv}_\perp},
  \label{eq:s_hat}
\end{equation}
where $\gv_\perp:= {\Um}_{\perp}^H\Gm\xv_p$.
The corresponding error variance is
\begin{equation}
  \sigma_{e_s}^2
  =
  \E\big[|s-\hat{s}|^2\big]
  =(\|\gv_\perp\|^2+\sigma^{-2})^{-1}
  \label{eq:s_mse}
\end{equation}
Conditioned on $\xv_p$, the estimator output satisfies
\[
  \hat{s}\mid \xv_p \sim \CN\!\big(0,\sigma^2-\sigma_{e_s}^2\big),
\]
and the estimation error $e_s:= s-\hat{s}$ is independent of $\hat{s}$. Note that $\|\gv_\perp\|^2$ depends on $\Um_\perp$ only via $\Scal(\Vm)$, so the estimation noise error depends only on $\Vm$. 

\subsection{Conditional LMMSE of $\wv_d$ given $(\hat{s},\xv_p)$}
\label{subsec:wd_cond_lmmse}

We decode the refined layer using the side information $\hat{s}$ from
\eqref{eq:s_hat}. First, we cancel the estimated random contribution of the
coarse layer:
\begin{equation}
  \tilde{\yv}_1
  :=
  \tilde{\yv}-\hat{s}\,\Gm\xv_p
  =
  (\Fm+\hat{s}\Gm)\Vm\wv_d
  + e_s\Gm(\xv_p + \Vm\wv_d)
  + \zv.
  \label{eq:y1_def}
\end{equation}
We then project onto the aligned refined-layer subspace $\Scal$.
\begin{equation}
  \yv_{\Scal}:= \Um^H\tilde{\yv}_1,
\end{equation}
where the projected Gaussian noise $\Um^H \zv$ is independent from the estimation noise $e_s$ that only depends on $\zv_\perp = \Um^H_{\perp} \zv$. 

Conditioned on $(\hat{s},\xv_p)$, we can compute the covariances $\Cm_{w y_{\Scal}\,|\,\hat{s},\xv_p}$ and $ \Cm_{y_{\Scal}y_{\Scal}\,|\,\hat{s},\xv_p}$ and
apply the LMMSE estimator
\begin{equation}
  \hat{\wv}_d(\hat{s}, \xv_p, \yv_{\Scal}) :=
  \Cm_{w y_{\Scal}\,|\,\hat{s},\xv_p}\,
  \Cm_{y_{\Scal}y_{\Scal}\,|\,\hat{s},\xv_p}^{-1}\,
  \yv_{\Scal},
  \label{eq:wd_cond_lmmse}
\end{equation}
where the conditional covariances depend on $\hat{s}$ through the effective
channel $(\Fm+\hat{s}\Gm)\Vm$ and on the disturbance induced by $e_s$.
Using standard methods, we can derive the corresponding conditional LMMSE covariance 
\small
\begin{equation}
  \Cm_e(\hat{s},\xv_p)
  :=
  \Big(
    \Qm_d^{-1}
    +
    \Vm^H(\Fm+\hat{s}\Gm)^H \Um(\Rm^{-1})\Um^H
    (\Fm+\hat{s}\Gm)\Vm
  \Big)^{-1},
  \label{eq:Ce_def}
\end{equation}
with
\begin{equation}
  \Rm
  := 
  \Id_{N-d_\perp}
  +
  \Um^H
  \sigma_{e_s}^2\,\Gm
  (\xv_p\xv_p^H+\Vm\Qm_d\Vm^H)
  \Gm^H \Um.
\end{equation}
\normalsize

Although \eqref{eq:wd_cond_lmmse} is linear in $\yv_{\Scal}$ for fixed $\hat{s}$ and $\wv_d$, the
overall mapping $\tilde{\yv}\mapsto\hat{\wv}_d$ is generally not linear,
since the filter coefficients depend on the random variable $\hat{s}$.
Equivalently, the receiver employs a random linear filter indexed by $\hat{s}$.

\subsection{Achievable rate}
\label{subsec:gmi_sideinfo}

We now apply the mismatched decoder with the nearest-neighbor metric as in Section~\ref{subsec:gaussian_lmmse}. Hence, the achievable rate for the refined layer,
\begin{equation}
  I(\wv_d;\yv\mid\xv_p,\hat{s}) = h(\wv_d) - h\!\left(\wv_d-\hat{\wv}_d\mid\yv,\xv_p,\hat{s}\right),
\end{equation}
is lower bounded by $\log\det(\Qm_d) - \E_{\hat{s},\xv_p}\!\left[ \log\det\!\big(\Cm_e(\hat{s},\xv_p)\big)\right]$. 
Recall that $\xv_p = \vv_p w_p$ with $w_p\sim \P_{w_p}$, we have the following achievable rate $ R_d(\P_{w_p},\Qm_d,\Vm,\vv_p)$: 
\begin{align}
   \E_{\hat{s},\xv_p}\!\biggl[
    &\log\det\Big(
    \Id\nonumber\\
    &\hspace{-0.5em}+ \Qm_d \Vm^H(\Fm+\hat{s}\Gm)^H 
    \Um(\Rm^{-1})\Um^H (\Fm+\hat{s}\Gm)\Vm  \Big)
  \biggr],
  \label{eq:Rd_cond_bound}
\end{align}
\normalsize
where the expectation is taken over the joint law of
$(\xv_p,\hat{s})$.

Combining the lower bounds on both layers, we get an achievable rate on the subspace alignment scheme. 

\begin{proposition}[Superposition coding with subspace alignment]
\label{prop:sa_achievable}
The channel capacity satisfies
\begin{equation}
  C(P)
  \;\ge\;
  \max_{\P_{w_p},\,\Vm,\,\vv_p,\,\Qm_d}
  R_{\mathrm{SA}}(\P_{w_p}, \Vm, \vv_p, \Qm_d),
  \label{eq:joint_opt}
\end{equation}
subject to the power constraints and to the alignment constraint
$\Vm \in \mathcal{V}(\Fm,\Gm)$, where
\begin{equation*}
  R_{\mathrm{SA}} := R_p + R_d.
\end{equation*}
\end{proposition}

A particularly relevant special case arises when $\P_{w_p}$ is chosen to be a Dirac measure, so that $\xv_p = \vv_p w_p$ becomes deterministic. In this situation, the achievable rate reduces to $R_{\mathrm{SA}} = R_d(\Qm_d,\Vm,\vv_p)$, because the coarse layer does not convey any information. Under these conditions, the proposed strategy effectively specializes to a purely pilot-based transmission scheme. 
In many parameter regimes, the use of a deterministic pilot entails only a negligible rate loss compared to superposition coding, since the primary role of the coarse layer is to facilitate reliable estimation of the state variable $s$.

\section{Capacity upper bound}
\label{sec:cap_ub}

We upper bound the mutual information using the duality (variational) bound
\begin{align}
  I(\xv;\yv)
  &= h(\yv)-h(\yv\mid\xv) \nonumber\\
  &\le \E\!\left[\log\frac{1}{q(\yv)}\right] - h(\yv\mid\xv),
  \label{eq:dual_start}
\end{align}
valid for any pdf $q(\yv)$, where the expectation is with respect to the true
output distribution $P_Y$ induced by $P_XP_{Y|X}$.

A standard approach is to optimize~\eqref{eq:dual_start} over a smoothly
parameterized family $\{q_\theta,\theta\in\Theta\}$. In this work, we consider
the regularized Gamma family from Lapidoth--Moser~\cite{LapidothMoser2003Duality}:
\begin{equation}
  q_{\theta}(\yv)
  =
  c_{\theta}\,
  \big(\|\yv\|_{\pmb S}^2 + \delta\big)^{\alpha - 1}\,
  \|\yv\|_{\pmb S}^{2(1-N)}\,
  e^{-({\|\yv\|_{\pmb S}^2 + \delta})/{\beta}},
  \label{eq:qtheta}
\end{equation}
where $\theta := (\alpha>0,\beta>0,\delta\ge 0,\pmb S\succ 0)$,
$\|\yv\|_{\pmb S}^2 := \yv^H\pmb S\,\yv$, and
$ c_{\theta}
  :=
  \frac{\Gamma(N)\det(\pmb S)}{\pi^N \beta^\alpha \Gamma(\alpha,\delta/\beta)},
  \Gamma(\alpha,\xi) := \int_{\xi}^{\infty} t^{\alpha-1}e^{-t}\,dt.
$
For convenience, we restrict the input to have bounded support
$\Xcal(r)=\{\xv:\|\xv\|\le r\}$. This ensures numerical stability and is without
loss of generality, since the resulting upper bound depends on $r$ and can be
made arbitrarily large. We further set $\delta=0$ and assume $\alpha\le N$, and
choose 
  $\beta := \beta(\Qm_x,\Sm)\;=\; \frac{1}{\alpha}\,\E\!\left[\|\yv\|_{\pmb S}^2\right]$.
\NEW{
Substituting the choice of $q_\theta$, the differential cross-entropy term in~\eqref{eq:dual_start} is written as
\begin{multline}
    \E\!\left[\log\frac{1}{q(\yv)}\right]
    = \alpha\log\beta(\Qm_x, \Sm) + \log \frac{\Gamma(\alpha)}{\Gamma(N)} + \alpha + N \log \pi
    \\[0.2em]
    - \log \det (\Sm) + (N-\alpha) \E \left[ \log \big( \|\yv\|_{\pmb S}^2\big) \big) \right]. \label{eq:diff_cross_entrop_q}
\end{multline}
}
Applying Jensen’s inequality to the conditional log-moment over $(s,\zv)$ yields
\begin{equation}
  \E_{s,\zv}\!\left[\log\!\big(\|\yv\|_{\pmb S}^2\big)\,\middle|\,\xv\right]
  \le
  \log\!\left(\E_{s,\zv}\!\left[\|\yv\|_{\pmb S}^2\,\middle|\,\xv\right]\right),
  \label{eq:ineq_log_squarred_y_giv_x}
\end{equation}
where
$  \E_{s,\zv}\!\left[\|\yv\|_{\pmb S}^2\,\middle|\,\xv\right]
  = \gamma(\Sm,\xv),$
with
\begin{equation}
  \gamma(\Sm,\xv)
  :=
  \tr(\pmb S)
  +\xv^H\!\big(\Fm^H\pmb S\Fm+\sigma^2\Gm^H\pmb S\Gm\big)\xv.
  \label{eq:cond_second_moment}
\end{equation}
Taking expectation over $\xv$ gives
\begin{equation}
  \beta(\Qm_x,\Sm)
  =
  \frac{1}{\alpha}
  \tr\!\Big(
    \Sm\big(\Id+\Fm\Qm_x\Fm^H+\sigma^2\Gm\Qm_x\Gm^H\big)
  \Big).
\end{equation}
\NEW{
Plugging \eqref{eq:ineq_log_squarred_y_giv_x} and \eqref{eq:cond_second_moment}  in \eqref{eq:diff_cross_entrop_q}, with \eqref{eq:h(y|x)} the bound in~\eqref{eq:dual_start} is written as
\begin{multline}
    I(\xv;\yv)
    \le \alpha\log\beta(\Qm_x, \Sm) + \log \frac{\Gamma(\alpha)}{\Gamma(N)} + \alpha - N 
    \\[0.2em]
    - \log \det (\Sm) + \E_{\xv} \left[ \log \frac{\gamma(\Sm, \xv)^{N-\alpha}}{1+\sigma^2 |\Gm\xv\|^2} \right],
    \label{eq:mut_inf_up_bound_x_S}
\end{multline}
}and upper bounding the expectation over $\xv$
by a supremum over $\Xcal(r)$ yields the following result.
\begin{proposition}
The channel capacity is upper bounded as
\begin{equation}
  C(P,r)
  \le
  \max_{\Qm_x\succeq0\atop\tr(\Qm_x)\le P}
  \inf_{\alpha\in(0,N]\atop \Sm\succeq0}
  R_{\mathrm{UB}}(\alpha,\Sm,\Qm_x),
\end{equation}
where
\begin{multline}
  R_{\mathrm{UB}}(\alpha,\Sm,\Qm_x)
  :=
  \alpha\log\beta(\Qm_x,\Sm)
  +\log\frac{\Gamma(\alpha)}{\Gamma(N)}
  +\alpha-N
  \\[0.2em]
  -\log\det(\pmb S)+\sup_{\xv\in\Xcal(r)}
  \log\!\left(
    \frac{\gamma(\Sm,\xv)^{\,N-\alpha}}
         {1+\sigma^2\|\Gm\xv\|^2}
  \right).
  \label{eq:dual_ub_ratio_general_alpha}
\end{multline}
\end{proposition}

Even without solving the minimization, valid upper bounds are obtained by fixing
$\alpha$ and $\Sm$. In particular, choosing $\alpha=N$ and
$\Sm=\det(\pmb\Sigma_{\yv})^{1/N}\pmb\Sigma_{\yv}^{-1}$, where
\begin{equation}
  \pmb\Sigma_{\yv}
  :=
  \Id+\Fm\Qm_x\Fm^H+\sigma^2\Gm\Qm_x\Gm^H,
\end{equation}
yields
\begin{equation}
  R_{\mathrm{UB}}
  = \alpha\log\beta
  = \log\det(\pmb\Sigma_{\yv}).
\end{equation}
Therefore,
\begin{equation}
  C(P)
  \le
  \max_{\Qm_x\succeq0\atop\tr(\Qm_x)\le P}
  \log\det\!\big(
    \Id+\Fm\Qm_x\Fm^H+\sigma^2\Gm\Qm_x\Gm^H
  \big),
  \label{eq:ub0}
\end{equation}
which is tight when $\sigma^2=0$.

\subsection{Small $\sigma$ regime}

When $\sigma$ is small, the trivial upper bound~\eqref{eq:ub0} yields
$
C(P) \le C_G(P) + O(\sigma^2),
$
where $C_G(P)$ denotes the capacity of the deterministic channel obtained for
$\sigma=0$. The $O(\sigma^2)$ term is positive and scales proportionally to
$\sigma^2$ as $\sigma\to 0$.

\NEW{
This $O(\sigma^2)$ term is obtained from~\eqref{eq:ub0} as follows
\begin{align}
    C(P) & \leq
    \max_{\Qm_x\succeq0\atop\tr(\Qm_x)\le P} \big( \log \det \pmb\Sigma_0
    \nonumber \\
    & \qquad \quad + \log \det (\Id + \sigma^2 \pmb\Sigma_0^{-1} \Gm \Qm_x \Gm^H) \big) \\
    & \leq C_G(P) \nonumber \\
    & \qquad \quad + \max_{\Qm_x\succeq0\atop\tr(\Qm_x)\le P} \log \det (\Id + \sigma^2 \pmb\Sigma_0^{-1} \Gm \Qm_x \Gm^H) \\ 
    & \leq C_G(P) + \max_{\Qm_x\succeq0\atop\tr(\Qm_x)\le P} \sigma^2 \tr (\pmb\Sigma_0^{-1} \Gm \Qm_x \Gm^H) \label{eq:sec_ord_mat_per_expan} \\
    & \leq C_G(P) + \sigma^2 f(P) \label{eq:sigma_fP} \\
    & \leq C_G(P) + O(\sigma^2,P)
\end{align}
where $\pmb\Sigma_0 = \Id + \Fm \Qm_x \Fm^H$ is the covariance of the deterministic channel, \eqref{eq:sec_ord_mat_per_expan} follows from
$\log\det(\Id + A)\le\tr(A)$ for $A\succeq0$ and the fact that the perturbation
term scales as $\sigma^2$, in~\eqref{eq:sigma_fP} we define
\begin{equation}
    f(P) = \max_{\Qm_x\succeq0\atop\tr(\Qm_x)\le P} 
    \tr (\pmb\Sigma_0^{-1} \Gm \Qm_x \Gm^H),
\end{equation}
where $f(P)$ is a function of the power constraint. If $f(P)$ is bounded, then the perturbation term is $O(\sigma^2)$ for all $P$.
}

On the other hand, from the achievable lower bound $R_G$ in~\eqref{eq:RG}, we
obtain
$
C(P) \ge C_G(P) + O(\sigma^2),
$
where the $O(\sigma^2)$ term is negative and of order $\sigma^2$. 
Indeed, since
$s$ is Gaussian, the perturbation
$
\Delta(s)
:=
(\Fm+s\Gm)\Qm_x(\Fm+s\Gm)^H - \Fm\Qm_x\Fm^H
$
exhibits exponentially decaying tails: as $\sigma\to 0$, the probability that
$\|\Delta(s)\|\ge 1$ is exponentially small. 
\NEW{
In this case, the $O(\sigma^2)$ term is obtained from \eqref{eq:gaussian_lb} as follows
\begin{align}
    C(P) & \ge \max_{\Qm_x\succeq0\atop\tr(\Qm_x)\le P}  \Big( \E_{s}\big[ \log \det (\Id + (\Fm + s \Gm)\Qm_x (\Fm + s \Gm)^H \big] \nonumber \\ 
    & \qquad \quad 
    - \log \big(1 + \sigma^2 \tr(\Gm \Qm_x \Gm^H) \big) \Big) \\
    & =  \max_{\Qm_x\succeq0\atop\tr(\Qm_x)\le P} \Big( \log\det \pmb\Sigma_0  + \E_{s} \big[ \log\det ( \Id + \pmb\Sigma_0^{-1} \Delta(s) )\big]   \nonumber \\ 
    & \qquad \quad 
    - \log \big(1 + \sigma^2 \tr(\Gm \Qm_x \Gm^H) \big) \Big). \label{eq:lb_sigma_0_sep}
\end{align}
Let $\Qm_x^* = \arg \max_{\Qm_x\succeq0\atop\tr(\Qm_x)\le P} \log \det \pmb\Sigma_0$, so that
\begin{equation}
    C_G(P) = \log \det \pmb\Sigma_0^*,
\end{equation}
and $\pmb\Sigma_0^*$ is the covariance of the deterministic channel evaluated at the optimal input covariance $\Qm_x^*$.
Then, from the second-order Taylor bound for the log-determinant, i.e., $\log \det (\Id + \Am) \ge \tr \Am - \frac{1}{2} \tr (\Am \Am)$ with $\|\Am\|\le 1$, and the first-order Taylor bound for the log, i.e., $\log (1+x) \le x$ with $|x| \le 1$, \eqref{eq:lb_sigma_0_sep} is lower bounded by
\begin{align}
    C(P)
    & \ge C_G(P) +  \E_{s} \big[ \tr({\pmb\Sigma_0^*}^{-1} \Delta(s)) \nonumber \\ 
    & \qquad \quad - \frac{1}{2} \tr({\pmb\Sigma_0^*}^{-1} \Delta(s) {\pmb\Sigma_0^*}^{-1} \Delta(s)) \big] - \sigma^2 \tr(\Gm \Qm_x^* \Gm^H ) \label{eq:lb_tr_ineq} \\
    & \ge C_G(P) + \sigma^2 \tr({\pmb\Sigma_0^*}^{-1} \Gm \Qm_x^* \Gm^H ) - \sigma^2  \| {\pmb\Sigma_0^*}^{-1} \Gm \Qm_x^* \Fm^H \|_F^2 \nonumber \\
    & \qquad \quad  - \sigma^2 \tr(\Gm \Qm_x^* \Gm^H) -\sigma^4 \tr(({\pmb\Sigma_0^*}^{-1} \Gm \Qm_x^* \Gm^H)^2) \\
    & \ge C_G(P) + O(\sigma^2),
\end{align}
where $\Delta(s) = s^* \Fm \Qm_x \Gm^H + s \Gm \Qm_x \Fm^H + |s|^2 \Gm \Qm_x \Gm^H$. Since $\mathbb{E}[s]=0$, all first-order terms vanish after averaging, and the leading correction is of order $\sigma^2$.
}

Therefore, in this regime, the capacity is indeed $C_G(P) + O(\sigma^2)$. 

\subsection{High SNR regime}

In the high-SNR regime $P\to\infty$, it follows from~\eqref{eq:RG} that
\[
R_G = (N-1)\log P + O(1),
\]
where the $O(1)$ term remains bounded as $P\to\infty$. Moreover, the pilot-based
scheme introduced in Section~IV is also able to achieve
$(N-1)\log P + O(1)$.

To establish a matching upper bound, we set $\alpha=N-1$ and $\Sm=\Id$ in
\eqref{eq:dual_ub_ratio_general_alpha}. After straightforward algebra, we obtain
the upper bound
\begin{multline}
  (N-1)\log\!\left(
    \frac{1}{N-1}\,\E\!\left[\gamma(\Id,\xv)\right]
  \right)
  + \log\frac{1}{N-1} - 1
  \\
  + \log \max\!\left\{
    N,\;
    \lambda_{\max}\!\left(
      \Id + \sigma^{-2}(\Gm^H\Gm)^{-1}\Fm^H\Fm
    \right)
  \right\},
\end{multline}
which scales as $(N-1)\log P + O(1)$, since
$\E[\gamma(\Id,\xv)]$ grows linearly with $P$.

Consequently, the channel has $N-1$ degrees of freedom. As in the small-$\sigma$
regime, obtaining tight expressions for the $O(1)$ terms in both the upper and
lower bounds remains an interesting direction for future work.

\section{Numerical Applications and Conclusion}

\begin{figure}
\includegraphics[width=\columnwidth]{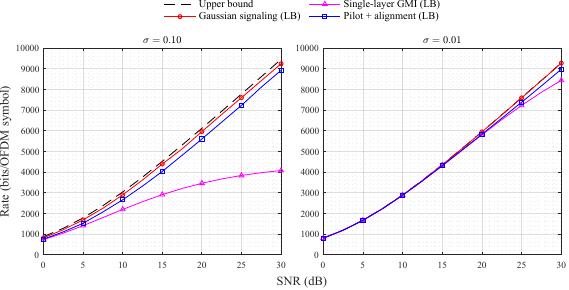}
\caption{Comparison of capacity bounds for Doppler-OFDM channel with $1024$ subcarriers. Left: $\sigma=0.1$. Right: $\sigma=0.01$.}
\label{fig:N1024}
\end{figure}

Numerical results are obtained using an OFDM-based NTN channel model. The carrier frequency is set to $f_c = 2$~GHz, the subcarrier
spacing to $\Delta f = 15$~kHz, and the system bandwidth to $20$~MHz. The sampling
period is $T_s = 1/(\Delta f \cdot 1024)$, corresponding to a standard OFDM
numerology. The multipath propagation follows the 3GPP NTN--TDL-A channel model~\cite{3gpp2020release15} in Table~\ref{tab1}.

\NEW{
The NTN-TDL-A model provides a tapped-delay-line (TDL) representation where each tap $k$ is characterized by a normalized delay $\tau_{k,\text{model}}$, an average tap power $P_k$ (in dB), and a fading distribution. The normalized delays are scaled to achieve a desired RMS delay spread $DS_{\text{desired}}$~\cite{3gpp2025release19} according to
\begin{equation}
\tau_k~[\text{ns}] = \tau_{k,\text{model}} \cdot DS_{\text{desired}}~[\text{ns}],
\end{equation}
with $DS_{\text{desired}}=100$ ns for the considered LEO satellite case. Then, the resulting continuous-time delays are mapped to discrete-time sample indices using the sampling period $T_s \approx 65.104$ ns. For each tap, the sample index can be computed as
\begin{equation}
n_k = \mathrm{round}\!\left(\frac{\tau_k}{T_s}\right).
\end{equation}

Using the normalized delays from Table~\ref{tab1}, the scaled delays are
$\tau_k \approx (0,\;108.11,\;284.16)$ ns, resulting in sample indices $n_k = (0,\;2,\;4)$.
The average tap powers are converted from dB to linear scale. Since each tap follows Rayleigh fading, the complex tap coefficient is generated as
\begin{equation}
\alpha_k = \sqrt{P_k}\, g_k, \qquad g_k \sim \mathcal{CN}(0,1),
\end{equation}
such that $\mathbb{E}\{|\alpha_k|^2\}=P_k$. Finally, the discrete-time channel impulse response is formed as
\begin{equation}
h_\ell = \sum_{k}^{L-1} \alpha_k\,\delta[\ell-n_k].
\end{equation}
}
\NEW{
\begin{table}[htbp]
\caption{NTN-TDL-A Channel Model}
\begin{center}
\begin{tabular}{|c|c|c|c|}
\hline
\textbf{Tap \#}& \textbf{Normalized Delay}& \textbf{Power (dB)}& \textbf{Fading Distribution}\\
\hline
1 & 0 & 0 & Rayleigh \\
\hline
2 & 1.0811 & -4.675 & Rayleigh \\
\hline
3 & 2.8416 & -6.482 & Rayleigh \\
\hline
\end{tabular}
\label{tab1}
\end{center}
\end{table}

The effective frequency-domain channel matrix is computed for one OFDM symbol and linearized with respect to the normalized Doppler shift, yielding the matrices $\Fm$ and $\Gm$ used in the abstract model $\Hm=\Fm+s\Gm$. The residual Doppler coefficient $s$ is modeled as a zero-mean complex Gaussian random variable with variance $\sigma^2$, constant over one OFDM symbol.
}

In Fig.\ref{fig:N1024}, we plot the capacity bounds derived in this work for moderate and small Doppler tracking error, $\sigma = 0.1$ and $\sigma=0.01$, respectively. We see that the Gaussian input lower bound with optimal decoding is close to the derived upper bound in this case. If the same input distribution is used but with suboptimal but practically used LMMSE receiver with nearest neighbor decoding metric, then we see that even the rank-one uncertainty can be detrimental and saturate the achievable rate even at moderate SNR. However, with the proposed superposition coding and subspace alignment scheme, we can achieve near optimal performance without sacrificing practicality. In this example, the coarse layer does not contain information and corresponds to the pilot based scheme. 

In summary, our information-theoretic analysis shows that the proposed subspace alignment scheme efficiently mitigates Doppler-induced inter-carrier interference and achieves near-capacity performance with low receiver complexity.

\newpage

\bibliographystyle{IEEEtran}
\bibliography{IEEEabrv,references}

\end{document}